\documentclass{article}
\usepackage{amsmath,mathrsfs,amsthm,amssymb,amsfonts,cite,amsopn,mathtools}
\usepackage{listings}
\usepackage{fancyvrb}
\usepackage[latin1]{inputenc}
\usepackage{colortbl}
\usepackage[english]{babel}
\usepackage{times}
\usepackage[normalem]{ulem}
\usepackage[titletoc,title]{appendix}
\usepackage{thmtools, thm-restate}
\usepackage{listings}
\usepackage{algorithm}
\usepackage{algorithmicx}
\usepackage{algpseudocode}
\algnewcommand{\LineComment}[1]{\Statex \(\triangleright\) #1}
\usepackage[hidelinks]{hyperref}
\usepackage{tabularx}
\usepackage{tikz}

\newlength{\defbaselineskip}
\usetikzlibrary{calc,shapes.arrows,chains,positioning}
\usepgflibrary{arrows} 
\usepgflibrary[arrows] 
\usepgflibrary[plotmarks] 
\usetikzlibrary{arrows} 

\DeclareMathOperator{\poly}{poly}

\DeclareMathOperator{\opt}{OPT}

\DeclareMathOperator{\dom}{domain}

\DeclarePairedDelimiter\ceil{\lceil}{\rceil}
\DeclarePairedDelimiter\dotp{\langle}{\rangle}

\newcommand\myeq{\mathrel{\overset{\makebox[0pt]{\mbox{\normalfont\tiny\sffamily def}}}{=}}}
\newcommand{\maxAi}{\norm[\infty]{A_{:i}}}
\newcommand{\deli}[1]{\nabla_i f_\mu(#1)}
\newcommand{\newx}{x^{(t)}_{k+1}}

\newtheorem{theorem}{Theorem}[section]

\newtheorem{claim}[theorem]{Claim}

\newtheorem{lemma}[theorem]{Lemma}

\theoremstyle{definition}

\newcommand{\argmax}{\operatornamewithlimits{argmax}}
\newcommand{\argmin}{\operatornamewithlimits{argmin}}

\newcommand{\R}{\mathbb{R}}
\newcommand{\E}{\mathbb{E}}

\newcommand{\vone}{\vec{1}}

\newcommand{\norm}[2][]{\ensuremath{\Vert #2 \Vert_{#1}}}

\begin{document}
\title{Faster Parallel Solver for Positive Linear Programs via \\ Dynamically-Bucketed Selective Coordinate Descent}

\date{}
\author{
  Di Wang
  \thanks{
  Department of Electrical Engineering and Computer Sciences,
  University of California at Berkeley,
  Berkeley, CA 94720.
  Email: wangd@eecs.berkeley.edu
  }
  \and
  Michael W. Mahoney
  \thanks{
  International Computer Science Institute
  and Department of Statistics,
  University of California at Berkeley,
  Berkeley, CA 94720.
  Email:  mmahoney@stat.berkeley.edu
  }
  \and
  Nishanth Mohan
  \thanks{
  Department of Electrical Engineering and Computer Sciences,
  University of California at Berkeley,
  Berkeley, CA 94720.
  Email: nishanth.mohan@berkeley.edu
  }
  \and
  Satish Rao
  \thanks{
  Department of Electrical Engineering and Computer Sciences,
  University of California at Berkeley,
  Berkeley, CA 94720.
  Email: satishr@berkeley.edu
  }
}

\maketitle

\begin{abstract}
We provide improved parallel approximation algorithms for the important class of packing and covering linear programs. 
In particular, we present new parallel $\epsilon$-approximate packing and covering solvers which run in $\tilde{O}(1/\epsilon^2)$ expected time, i.e., in expectation they take $\tilde{O}(1/\epsilon^2)$ iterations and they do $\tilde{O}(N/\epsilon^2)$ total work, where $N$ is the size of the constraint matrix and $\epsilon$ is the error parameter, and where the $\tilde{O}$ hides logarithmic factors.
To achieve our improvement, we introduce an algorithmic technique of broader interest: \emph{dynamically-bucketed selective coordinate descent (DB-SCD)}.  
At each step of the iterative optimization algorithm, the DB-SCD method dynamically buckets the coordinates of the gradient into those of roughly equal magnitude, and it updates all the coordinates in one of the buckets.  
This dynamically-bucketed updating permits us to take steps along several coordinates with similar-sized gradients, thereby permitting more appropriate step sizes at each step of the algorithm. 
In particular, this technique allows us to use in a straightforward manner the recent analysis from the breakthrough results of Allen-Zhu and Orecchia~\cite{AO15-parallel} to achieve our still-further improved bounds.
More generally, this method addresses ``interference'' among coordinates, by which we mean the impact of the update of one coordinate on the gradients of other coordinates. 
Such interference is a core issue in parallelizing optimization routines that rely on smoothness properties. 
Since our DB-SCD method reduces interference via updating a selective subset of variables at each iteration, we expect it may also have more general applicability in optimization.
\end{abstract}

\section{Introduction}
\label{sxn:intro}

Packing and covering problems are important classes of linear programs with many applications, and they have long drawn interest in computer science in general and theoretical computer science in particular. 
In their generic form, fractional packing problems can be written as the linear program (LP): 
$$\max_{x\geq 0}\{c^Tx:Ax\leq b\},$$ 
where $c\in \R^n_{\geq 0},b\in \R^m_{\geq 0}$, and $A\in \R^{m\times n}_{\geq 0}$ are all non-negative. 
Without loss of generality, one can scale the coefficients, in which case one can write this LP in the standard form:
\begin{equation}
\label{eq:packLP}
\max_{x\geq 0}\{\vone^Tx:Ax\leq \vone\}  ,
\end{equation}
where $A\in \R^{m\times n}_{\geq 0}$.
The dual of this LP, the fractional covering problem, can be written in the standard form as:
\begin{equation}
\label{eq:covLP}
\min_{y\geq 0}\{\vone^Ty:A^Ty\geq \vone\} .
\end{equation}
We denote by $\opt$ the optimal value of the primal problem~\eqref{eq:packLP} (which is also the optimal value of the dual problem~\eqref{eq:covLP}).
In this case, we say that a vector $x$ is a {\it $(1-\epsilon)$-approximation} for the packing LP if $Ax\leq \vone$ and $\vone^Tx\geq (1-\epsilon)\opt$, and we say that $y$ is a {\it $(1+\epsilon)$-approximation} for the covering LP if $Ay\geq \vone$ and $\vone^Ty\leq (1+\epsilon)\opt$.

In this paper, we describe improved parallel algorithms for packing LPs and covering LPs, improving the dependence on the error parameter $\epsilon$ from $\tilde{O}(1/\epsilon^3)$ to $\tilde{O}(1/\epsilon^2)$ for both the total work and the distributed iteration count for both problems~\eqref{eq:packLP} and~\eqref{eq:covLP}.
Our approach follows the general approach of transforming non-smooth LPs to smooth convex optimization problems and then applying an efficient first-order optimization algorithm. 
Unfortunately, the smoothed objective that arises does \emph{not} have particularly Lipschitz continuity properties, and thus we are unable to use traditional optimization methods to improve the (parallel) convergence rate beyond $\tilde{O}(1/\epsilon^3)$.
Thus, to achieve our improvement to $\tilde{O}(1/\epsilon^2)$, we develop the \emph{dynamically-bucketed selective coordinate descent (DB-SCD)} method. 
This descent method involves partitioning coordinates into buckets based on the magnitudes of their gradients and updating the coordinates in one of the buckets.
This permits us to make relatively large gradient moves along a subset of the coordinates for which we can control the smoothness of gradients within the range of our step. 
Given that controlling the smoothness properties of functions is central to controlling the convergence rate of continuous optimization algorithms, we expect that this method will be useful more~generally.

\subsection{Overview of Prior Methods}

Although one can use general LP solvers such as interior point methods to solve packing and covering problems with a convergence rate of $O(\log(1/\epsilon))$, such algorithms usually have very high per-iteration cost, as methods such as the computation of the Hessian and matrix inversion are involved. 
In the setting of large-scale problems, low precision iterative solvers are often more popular choices. 
Such solvers usually run in time with a much better dependence on the problem size, but they have the much worse $\poly(1/\epsilon)$ dependence on the approximation~parameter. 
Most such work falls into one of two categories. 
The first category follows the approach of transforming LPs to convex optimization problems, then applying efficient first-order optimization algorithms. 
Examples of work in this category include~\cite{Nemirovski04,AHK12,Nesterov05,Renegar14,AO15-parallel,AO15-stochastic}, and all except~\cite{AO15-parallel,AO15-stochastic} apply to more general classes of LPs. 
The second category is based on the Lagrangian relaxation framework, and some examples of work in this category include~\cite{PST91,Fleischer04,LubyN93,Young01,Young14,KY14}.
For a more detailed comparison of this prior work, see Table $1$ in ~\cite{AO15-stochastic}.
Based on whether the running time depends on the width $\rho$, a parameter which typically depends on the dimension or the largest entry of $A$, these algorithms can also be divided into width-dependent solvers and width-independent solvers. 
Width-dependent solvers are usually pseudo-polynomial, as the running time depends on $\rho\opt$, which itself can be large, while width-independent solvers are more efficient, in the sense that they provide truly polynomial-time approximation solvers. 

The line of research associated with width-independent solvers was initiated by Luby and Nisan~\cite{LubyN93}, where the authors gave a parallel algorithm runs in $\tilde{O}\left(1/\epsilon^4\right)$ distributed iterations and  $\tilde{O}\left(N/\epsilon^4\right)$ total work. 
Note that, since we are most interested here in the dependence on the error parameter $\epsilon$, to simplify the discussion and notation, we will follow the standard practice of using $\tilde{O}$ to hide poly-log factors. 
For readers interested in the more precise results, see Table~\ref{tab:running-times} in this section as well as our analysis below. 
For sequential algorithms, on the total work front, a recent breakthrough gives an $\tilde{O}(N/\epsilon)$ sequential algorithm for packing and a different $\tilde{O}(N/\epsilon^{1.5})$ sequential algorithm for covering~\cite{AO15-stochastic}.
Our recent work improved this by developing a diameter reduction method that leads to a unified framework that achieves an $\tilde{O}(N/\epsilon)$ sequential algorithm for both packing and covering~\cite{WRM15}. 
In terms of parallel algorithms, improvement over the $\tilde{O}(1/\epsilon^4)$ iteration count and $\tilde{O}(N/\epsilon^4)$ total work in the original paper of Luby and Nisan~\cite{LubyN93} was only achieved recently. 
In particular, Allen-Zhu and Orecchia~\cite{AO15-parallel} gave a deterministic algorithm with $\tilde{O}(1/\epsilon^3)$ iterations and $\tilde{O}(N/\epsilon^3)$ total work. 

\begin{table}
\begin{center}
\begin{tabular}{l|c|c|}
 & Number of distributed iterations & Total work  \\
\hline \hline
Luby and Nisan~\cite{LubyN93} & $O\left(\frac{\log^2 N}{\epsilon^4}\right)$ & $O\left(\frac{\log^2 N}{\epsilon^4}\times (N\log n)\right)$   \\
Allen-Zhu and Orecchia~\cite{AO15-parallel} & $O\left(\frac{\log^2 N}{\epsilon^3}\right)$ & $O\left(\frac{\log^2 N}{\epsilon^3}\times N\right)$  \\
Our main results   & $O\left(\frac{\log^2 N\log \frac{1}{\epsilon}}{\epsilon^2}\right)$ & $O\left(\frac{\log^2 N\log \frac{1}{\epsilon}}{\epsilon^2}\times N \right)$  \\
\end{tabular}
\caption{Running time for several parallel solvers for packing and covering LP problems. $N$ is the total number of non-zero elements in the constraint matrix. Our algorithm is randomized, so the number of distributed iterations and total work are in expectation.}
\end{center}
\label{tab:running-times}
\end{table}

\subsection{Our Main Results}
\label{sxn:intro-main}

In this paper, we describe improved parallel algorithms for packing LPs and for covering LPs, improving the dependence on the error parameter $\epsilon$ for total both work and distributed iteration count for both problems~\eqref{eq:packLP} and~\eqref{eq:covLP} from $\tilde{O}(1/\epsilon^3)$ to $\tilde{O}(1/\epsilon^2)$.
In particular, we present a stochastic parallel solver that provides a $(1-\epsilon)$-approximation for primal packing LPs of the form~\eqref{eq:packLP}.
The solver is a width-independent first-order method.
It is a stochastic method, and it converges not only in expectation, but also with at least constant probability. 
Furthermore, our solver has the additional guarantee that the objective function value is non-increasing across iterations. In general, stochastic first order methods, such as coordinate descent and stochastic gradient descent, show the expectation of the objective function converges to optimum, without the monotonicity guarantee on the actual objective (e.g.,~\cite{LiuWRBS14,FR13a,RichtarikT14,Nesterov12}). In practice, when the constraints in the problem are ill-conditioned or highly non-separable, the objective function value may fluctuate heavily during execution of stochastic methods, and this has motivated the development of more robust stochastic algorithms (e.g.,~\cite{Kimon14}). 

More precisely, here is our main theorem for the fractional packing problem.
\begin{theorem}
\label{thm:main2}
There is a randomized algorithm that, with probability at least $9/10$, computes a $(1-O(\epsilon))$-approximation to the fractional packing problem, has 
$\tilde{O}(N/\epsilon^2)$ total work, and is implementable in 
$\tilde{O}(1/\epsilon^2)$ 
distributed~iterations.
\end{theorem}

\noindent
In addition to this result for the primal packing problem, as given by~\eqref{eq:packLP}, we can use the primal packing solver to get a $(1+\epsilon)$-approximation to the dual covering problem, as given by~\eqref{eq:covLP}.
Here is our main theorem for the fractional covering~problem.
\begin{theorem}
\label{thm:dual2}
There is a randomized algorithm that, with probability at least $9/10$, computes a $(1+O(\epsilon))$-approximation to the fractional covering problem, has 
$\tilde{O}(N/\epsilon^2)$ total work, and is implementable in 
$\tilde{O}(1/\epsilon^2)$ 
distributed~iterations.
\end{theorem}

That is, our packing solver and our covering solver have $\tilde{O}(1/\epsilon^2)$ expected iterations and $\tilde{O}(N/\epsilon^2)$ expected total work. 
Among other things, this gives an expected improvement of $\tilde{O}(1/\epsilon)$ over the current fastest parallel algorithm of Allen-Zhu and Orecchia~\cite{AO15-parallel}, in terms of both iteration count and total work.
See Table~\ref{tab:running-times} for a more detailed comparison with their results as well as the results of Luby and Nisan~\cite{LubyN93}.
See also Section~\ref{sxn:main-results}, which contains a more detailed statement of these results as well as the algorithms that achieve these results.

\subsection{Our Main Techniques}
\label{sxn:intro-techniques}
 
The general approach of transforming non-smooth LPs into smooth convex optimization problems and then applying efficient first-order optimization algorithm has been done by Nesterov~\cite{Nesterov05}, Allen-Zhu and Orecchia~\cite{AO15-parallel,AO15-stochastic}, and many others.
In particular, following~\cite{AO15-parallel,AO15-stochastic}, to find a $(1-\epsilon)$-approximation of a packing LP, we will approximately minimize, over the region $x\geq \vec{0}$, the following convex function:
\begin{equation}
\label{eqn:smoothing}
f_\mu(x)=-\vone^Tx+\mu \sum_{j=1}^m \exp(\frac{1}{\mu}((Ax)_j)-1) .
\end{equation}
In Section~\ref{sxn:solver-prelim}, we will discuss the motivation of using $f_\mu(x)$ to solve our packing LP as well as properties of the parameter $\mu$ and the function $f_\mu(x)$. 
Here, we will focus on the main techniques we had to introduce in order to obtain our improved solver.


To do so, recall that first-order methods in optimization exploit the idea of using the negative of the gradient as the descent direction for each step of an iterative algorithm.
In order to lower bound the improvement of successive steps of the algorithm, smoothness is at the core of the analysis.
The reason is basically since we want to move in the descent direction as far as possible, without changing the gradient by too much. 
This is most commonly captured by proving the Lipschitz continuity property with parameter $L$ on the function $f$, i.e., we want to find an $L\in\mathbb{R}^{+}$ such that
\begin{equation}
\label{eq:LipschitzDef}
\norm{\nabla f(x)-\nabla f(y)}^*\leq L\norm{x-y} \quad \forall x,y\in \dom(f)   .
\end{equation}
While popular, the Lipschitz continuity property is often much stronger than necessary.
For example, instead of controlling the properties of the gradient for all $x,y$ pairs, in most gradient based methods, we actually only need to show the smoothness of gradients within the range of our step---which, by design, we can control (indeed, step sizes are often chosen based on this). 
This motivates the use of weaker continuity properties by focusing on more constrained cases of $x,y$ pairs that are still sufficient to lower bound the improvement of successive steps of the algorithm.

Our basic improvement with the DB-SCD method comes from updating only a carefully-chosen subset of variables in each iteration. 
In particular, we bucket the coordinates of the gradient into buckets of roughly equal gradient magnitude, according to~\eqref{eq:groups} below, and then we update the coordinates in one of the buckets, according to Step~\ref{step:selective-update} of Algorithm~\ref{alg:PP} below.
While our particular approach is novel, the basic idea is hardly new to optimization. 
Indeed, for most non-trivial functions, variables ``interfere'' with each other, in the sense that variable $i$'s gradient will be affected by the update of variable $j$ (e.g.,~\cite{BradleyKBG11}). 
Thus, if we aim to move the variables while maintaining smoothness of the gradients, we have to take interference into consideration. 
In general, this limits the possible step size. The global Lipschitz continuity parameter usually suffers from interference, since the Lipschitz property~\eqref{eq:LipschitzDef} needs to hold for all $x,y$ pairs, which allows arbitrary interference. 

One way to alleviate the problem of interference is to update fewer variables in each iteration.
This potentially permits better control over the changes of gradients for the updated variables, since for the variables not updated, the changes of their gradients don't affect the objective improvement for that iteration. 
One extreme of this idea is the coordinate descent method~\cite{Wright15_CD}, where in each iteration only one variable is updated. 
In this case, the step length of the update on the single variable is often larger than the step length when all variables are moved simultaneously. 
On the other hand, in most cases the computation of $n$ successive partial derivatives can be more expensive than the computation of all the $n$ partial derivatives for a fixed $x$, limiting the applicability of the coordinate descent method. 
When the tradeoff is good between the gain in the step length versus the loss in the computation, coordinate descent can be better than gradient descent in terms of total work (\cite{LeeS13_CD,AO15-stochastic}). 
More generally, in the context of solving linear systems, we have the example of Jacobi iterations versus Gauss-Seidel iterations, and similar tradeoffs between interference and running time arise (\cite{Amodio:1995,BertsekasT91}, Chapter $4$ of~\cite{saad2003iterative}).
Still more generally, various efforts to parallelize coordinate descent can be seen as explorations of tradeoffs among the smoothness parameter, the amount of computation, as well as the distributed iteration count (\cite{FR13a,RichtarikT12,RichtarikT14,BradleyKBG11}). 

To the best of our knowledge, all such works mentioned exhibit an inverse relationship between the number of variables updated each iteration, and the number of total iterations required, i.e., when fewer variables are updated, then more iterative steps are needed. 
This is what one would naturally expect. 
Moreover, the prior works mentioned mostly choose the subset of variables to update either by some fix order, e.g., uniformly at random, or according to some static partition constructed from the sparsity structure of the given instance, e.g., the objective function is separable or the matrix in the problem is block diagonal (~\cite{TsengY09}).
Rarely, if at all, is a subset of variables chosen dynamically using knowledge of the actual values of the gradients. 
Again, this is what one would naturally expect. 
For example, as Nesterov wrote in his seminal accelerated coordinate descent work~\cite{Nesterov12}, if one already computed the whole gradient, then full-gradient methods seem to be better options. 

With respect to both of these considerations, our method of selective coordinate descent is novel and quite different.
First, at least for the case of packing and covering LPs, we can achieve better parallel running time and better total work by updating fewer (carefully-selected) variables each iteration. 
Second, our work shows that the extra computation of the whole gradient can help us select a better subset dynamically (even if we don't update all coordinates).
Our results in this paper show that both of these directions are worth additional exploration. 
Finally, we emphasize that a less obvious benefit of our approach is that the gradients in most cases contain useful information about the dual problem, and if we have the whole gradients from all iterations, then we can exploit the primal-dual structure of the packing problem to obtain a solution to the covering problem.


\section{Faster Parallel Solver for Packing LPs and Covering LPs}
\label{sxn:main-results}

In this section, we will present our main results, including a statement of our main algorithm and theorem for a parallel solver for packing LPs (in Section~\ref{sxn:main-results-packing}), a statement of our main algorithm and theorem for a parallel solver for covering LPs (in Section~\ref{sxn:main-results-covering}), and a description of the main technical ideas underlying our proofs (in Section~\ref{sxn:main-results-ideas}).

\subsection{Algorithm and Main Theorems for Parallel Packing LP Solver}
\label{sxn:main-results-packing}

Our main algorithm to find a $(1-\epsilon)$-approximation of a packing LPs is specified in Algorithm~\ref{alg:PP}. 
This algorithm takes as input the matrix $A$ as in~\eqref{eq:packLP}, the smoothed objective function $f_\mu$, and the approximation parameter $\epsilon$. It returns as output $x_T$, such that with constant probability $\frac{1}{1+\epsilon}x_T$ is a $(1-O(\epsilon))$-approximation to the packing problem.
In this algorithm, we use $x[i]$ to denote the $i$-th coordinate of vector $x$, except with matrix-vector multiplications, where we use $(Ax)_i$ to denote the value of the $i$-th component of $Ax$.

\begin{algorithm}
\caption{Stochastic and Parallelizable Packing LP Solver}
\label{alg:PP}
\textbf{Input}: $A\in \R^{m\times n}_{\geq 0},f_\mu,\epsilon\in (0,1/2]$
\textbf{Output}: $x\in \R^n_{\geq 0}$
\begin{algorithmic}[1]
\State $\mu\leftarrow\frac{\epsilon}{4\log(nm/\epsilon)}, \alpha \leftarrow\frac{\mu}{20}$
\State $T\leftarrow\frac{\ceil*{10\log(1/\epsilon)\log(2n)}}{\alpha\epsilon}=\tilde{O}(\frac{1}{\epsilon^2}), w\leftarrow \ceil*{\log(\frac{1}{\epsilon})}$
\State $x_0\leftarrow \frac{1-\epsilon/2}{n\maxAi}$
\For{$k=0$ to $T-1$}
\State Select $t\in \{0,\ldots,w-1 \}$ uniformly at random
\For{$i=1$ to $n$}
\State Gradient truncation and coordinate selection: Compute $\nabla_i f_\mu(x_k)$, and get $\xi^{(t)}_k[i]$, as defined in~\eqref{eq:groups}
\label{step:selective-update}
\State Update step: $x_{k+1}[i]\leftarrow \newx[i]\myeq x_k[i]\exp(-\alpha \xi^{(t)}_k[i])$
\EndFor
\EndFor
\State \Return $x_T$.
\end{algorithmic}
\end{algorithm}

To understand the steps of Algorithm~\ref{alg:PP}, recall that the function $f_\mu(x)$ referred to in the algorithm is given by Eqn.~(\ref{eqn:smoothing}) and is a smoothed version of the packing objective of~\eqref{eq:packLP}.
Then, following~\cite{AO15-parallel}, in each iteration $k \in[0,T-1]$, for each variable $x_k[i]$, we can break the gradients $\deli{x_k}$ into small, medium, and large components.
That is, we can let
%
%
\begin{eqnarray}
\zeta_k[i] &=& \left\{\begin{array}{lr}
\nabla_if_\mu(x_k) & \nabla_if_\mu(x_k)\in [-\epsilon,\epsilon]\\
0 & \text{otherwise} \end{array} \right.
\\
\xi_k[i] &=& \left\{\begin{array}{lr}
0 & \nabla_if_\mu(x_k)\in [-\epsilon,\epsilon]\\
\nabla_if_\mu(x_k) & \nabla_if_\mu(x_k)\in [-1,1]\backslash [-\epsilon,\epsilon]\\
1 & \nabla_if_\mu(x_k)>1 \end{array} \right.
\\
\eta_k[i] &=& \left\{\begin{array}{lr}
\nabla_if_\mu(x_k)-1 & \nabla_if_\mu(x_k)>1\\
0 & \text{otherwise} \end{array} \right.
\end{eqnarray}
denote, respectively, the small, medium, and large components of the gradient.
In particular, from this decomposition, we have
\begin{equation}
\label{eq:smallLarge}
\nabla f_\mu(x_k)=\zeta_k+\xi_k+\eta_k  .
\end{equation}
(It was by adopting this partitioning that previous work achieved their $\tilde{O}(1/\epsilon^3)$ running time~\cite{AO15-parallel}.)

In Lemma~\ref{lemma:Lipschitz} below, we will establish that if the gradients are all within a factor of (say) $2$ from each other, then we can take a multiplicative step with step size $\alpha=\Theta(\mu)=\tilde{O}(\epsilon)$. 
To exploit this algorithmically, and to lead to our improved algorithmic results, we will further partition the variables into groups such that, for variables in the same group, the absolute values of their truncated gradients will be within a factor $2$ of each other. 
In particular, we will further partition the medium component into groups or buckets in which the truncated gradients are of roughly equal magnnitude.
To do this, for $t\in\{0,\ldots,\ceil{\log(\frac{1}{\epsilon})-1}\}$, we let
\begin{equation}
\label{eq:groups}
\xi^{(t)}_k[i]= \left\{\begin{array}{lr}
\xi_k[i] & \xi_k[i] \in (\epsilon 2^t,\epsilon 2^{t+1}]\\
& \cup [-\epsilon 2^{t+1},-\epsilon 2^t)\\
0 & \text{otherwise} \end{array} \right.\quad\mbox{and}\quad
\eta^{(t)}_k[i] = \left\{\begin{array}{lr}
\eta_k[i] & t=\ceil{\log(\frac{1}{\epsilon})}-1\\
0 & \text{otherwise} \end{array} \right.  .
\end{equation}
Then, in each iteration of Algorithm~\ref{alg:PP}, we will pick a bucket $t$ uniformly at random, and we will update all variables using $\xi_k^{(t)}$. 

Our main result for Algorithm~\ref{alg:PP} is summarized in the following theorem, the proof of which we will present in Section~\ref{sxn:analysis}.

\begin{theorem}
\label{thm:main}
Algorithm~\ref{alg:PP} outputs $x_T$ satisfying $\E[f_\mu(x_T)]\leq -(1-5\epsilon)\opt$, and the algorithm can be implemented with $O(\frac{\log(1/\epsilon)\log^2N}{\epsilon^2})$ iterations with total work $O(N\times \frac{\log(1/\epsilon)\log^2N}{\epsilon^2})$, where $N$ is the total number of non-zeros in $A$.
\end{theorem}

\noindent
\textbf{Remark.}
Our main Theorem~\ref{thm:main2} follows almost immediately from Theorem~\ref{thm:main}, when combined with a standard application of Markov bound and part $(5)$ of Lemma~\ref{lemma:smoothing_properties} below.
In particular, by Lemma~\ref{lemma:smoothing_properties}$(2)$, for every $x\geq 0$, $f_\mu(x)\geq -(1+\epsilon)\opt$. 
From Theorem~\ref{thm:main}, we have that $f_\mu(x_T)+(1+\epsilon)\opt$ is a non-negative random variable with expectation at most $4\epsilon$. 
Using Markov's inequality, with at least probability $9/10$, we have $f_\mu(x_T)\leq -(1-41\epsilon)\opt$, giving a $(1-O(\epsilon)$ approximation by Lemma~\ref{lemma:smoothing_properties}$(5)$. 
The total work and iteration count follow directly from Theorem~\ref{thm:main}, thus proving our main Theorem~\ref{thm:main2}.

\subsection{Algorithm and Main Theorems for Parallel Covering LP Solver}
\label{sxn:main-results-covering}

A benefit of computing all the gradients in each iteration of Algorithm~\ref{alg:PP}  is that---even if we don't use them for our packing solver---we can exploit the same primal-dual structure as in~\cite{AO15-parallel} to get a covering LP solver. 
In particular, given a covering LP instance in the form of~\eqref{eq:covLP}, we can construct its dual, which is a packing LP. 
If we then run Algorithm~\ref{alg:PP} on the packing instance for $T$ iterations, then the average of the exponential penalties used in the computation of gradients, i.e.,  
\begin{equation}
\label{eq:dualAve}
\bar{y}=\frac{1}{T}\sum_{k=0}^{T-1}\overrightarrow{p(x_k)}\geq 0  ,
\end{equation}
will, with constant probability, be a $(1+\epsilon)$-approximation of the covering problem, after some simple fixing step. The fixing step is to post-processing $\bar{y}$ to enforce feasibility of the covering solution, as $\bar{y}$ will only be feasible with high probability. Here $\overrightarrow{p(x_k)}$ is the vector of all the exponential penalties of the packing constraints, as defined in Lemma~\ref{lemma:explicitFunction}, which we compute in each iteration of Algorithm~\ref{alg:PP} to get the gradients.
To obtain the dual solution, the primal-dual property we will exploit is that the slackness of the $i$-th covering constraint with $\bar{y}$ is the average gradient of the $i$-th variable in the primal packing LP:
\[
(A^T\bar{y})_i-1=\frac{1}{T}\sum_{k=0}^{T-1}(A^T\overrightarrow{p(x_k)})_i-1=\frac{1}{T}\sum_{k=0}^{T-1}\deli{x_k}  .
\]
Following a similar approch as in~\cite{AO15-parallel,ZhuLO15a}, we present our fixing procedure in Algorithm~\ref{alg:fix}.
Note, in particular, that this explicitly makes all the dual constraints satisfied.

\begin{algorithm}
\caption{Post-processing $\bar{y}$ to enforce feasibility of the dual solution for the Covering LP solver}
\label{alg:fix}
\textbf{Input}: $A\in \R^{m\times n}_{\geq 0},\epsilon\in (0,1/10],\bar{y}\in R^{m}_{\geq 0}$
\textbf{Output}: $y\in \R^m_{\geq 0}$ such that $A^Ty\geq \vone$.
\begin{algorithmic}[1]
\State $\bar{y}'\leftarrow \bar{y}$
\ForAll{$i$ such that $\lambda_i\myeq (A^T\bar{y})_i-1+\epsilon\leq -2\epsilon$}
\State Let $j=\argmax_j' A_{i,j'}$, i.e. $A_{i,j}=\maxAi$.
\State $\bar{y}'_j\leftarrow \bar{y}'_j + \frac{-\lambda_i}{A_{i,j}}$.
\EndFor
\State \Return $y=\frac{\bar{y}'}{1-3\epsilon}$.
\end{algorithmic}
\end{algorithm}

Overall, given a covering LP instance in the form of~\eqref{eq:covLP}, the entire covering solver consists of the following.
\begin{itemize}
\item
First, construct its dual, which is a packing LP in the form of~\eqref{eq:packLP}. 
\item
Then, run Algorithm~\ref{alg:PP} for $T$ iterations, with $T\geq \max\{\frac{6w}{\alpha\epsilon},\frac{2w^2\log\frac{n}{\epsilon}}{\epsilon^2}\}$.
\item
Finally, fix the $\bar{y}$ as in~\eqref{eq:dualAve} with Algorithm~\ref{alg:fix}, which takes $O(\log(N))$ time and $O(N)$ work. 
\end{itemize}
If $y$ is the output of Algorithm~\ref{alg:fix}, then $y$ is feasible by construction, i.e., $y\geq 0, A^Ty\geq \vone$, and moreover we can establish the following result.

\begin{theorem}
\label{thm:dual}
$\E[\vone^Ty]\leq (1+10\epsilon)\opt$, and $y\geq 0, A^Ty\geq \vone$.
\end{theorem}

\noindent
\textbf{Remark.}
Our main Theorem~\ref{thm:dual2} follows almost immediately from Theorem~\ref{thm:dual}, when combined with a standard application of Markov bound and several lemmas below.
From Theorem~\ref{thm:dual}, we have $\E[\vone^Ty]\leq (1+10\epsilon)\opt$, and $\vone^Ty>\opt$, since $y$ is always feasible. 
Then $y-\opt$ is a non-negative random variable with expectation at most $10\epsilon$. 
Using Markov's inequality, with at least probability $9/10$, we have $\vone^Ty \leq (1+100\epsilon)\opt$, giving a $(1+O(\epsilon)$ approximation. 
The expected total work and iteration count is dominated by Algorithm~\ref{alg:PP}. 
We show in Lemma~\ref{lemma:bary} and Lemma~\ref{lemma:martingale} that runing Algorithm~\ref{alg:PP} for $T=\max\{\frac{6w}{\alpha\epsilon},\frac{2w^2\log\frac{n}{\epsilon}}{\epsilon^2}\}=O(\frac{\log^2(\frac{1}{\epsilon})\log N}{\epsilon^2})$ iterations is sufficient.
This proves our main Theorem~\ref{thm:dual2}.

\subsection{Discussion of Main Technical Ideas Underlying Our Proofs}
\label{sxn:main-results-ideas}

Before proceeding with our proofs of Theorems~\ref{thm:main} and~\ref{thm:dual} in Sections~\ref{sxn:analysis} and~\ref{sxn:proof-dual-thm}, repsectively, we provide here a discussion of the main technical ideas underlying our methods.

At a high level, we (as well as Allen-Zhu and Orecchia~\cite{AO15-parallel,AO15-stochastic}) use the same two-step approach of Nesterov~\cite{Nesterov05}. 
The first step involves smoothing, which transforms the constrained problem into a convex and {\it smooth} objective function with trivial or no constraints. 
(By smooth, we mean that the gradient of the objective function has some property in the flavor of Lipschitz continuity.) 
In our case, we optimize the function $f_\mu(x)$ with the trivial constraints $x\geq \vec{0}$, and the smoothness property of $f_\mu(x)$ is specified in Lemma~\ref{lemma:Lipschitz}. 
Once smoothing is accomplished, the second step uses one of several first order methods for convex optimization in order to obtain an approximate solution. 
Examples of standard application of this approach to packing and covering LPs includes the width-dependent solvers of~\cite{Nesterov05,Nemirovski04} as well as multiplicative weights update solvers~\cite{AHK12}. 
In these examples, the dependence of $\opt$ is introduced, when the entropy function is used in the smoothing step. The dependence on $\rho$ is from using the full gradient, which can be large if there are large entries in $A$.

We will see in Section~\ref{sxn:solver-prelim} that $f_\mu(x)$ is the objective function derived from a different smoothing step, which avoids the dependence of $\opt$. 
The function $f_\mu(x)$ is used in~\cite{AO15-parallel}, which is the first width-independent result following the optimization approach, as well as in later works~\cite{AO15-stochastic,WRM15}. 
The different smoothing alone is not enough for width-independence, however, since in the optimization step, the convergence of the first order method will depend on the largest absolute value of the gradient or feedback vector. 
We use the same idea as in~\cite{AO15-stochastic,AO15-parallel} to use truncated gradient in our algorithm, thus effectively reducing the width to $1$. 
More specifically, the regret term in standard mirror descent analysis depends on the width, while in Lemma~\ref{lemma:mirror}, our regret term $\alpha^2\opt$ doesn't.

A major issue that arises with this approach is that the convex objective function $f_\mu(x)$ is \emph{not} smooth in the standard Lipschitz continuity sense. 
The authors in~\cite{AO15-parallel} work around it by showing that the gradient is multiplicatively smooth within some local neighborhood, and they constrain their update steps to stay inside that region. 
The bottleneck of their convergence, as we see it, is that the step size of the update is too small due to interference (recall that interference is the dependence of variable $i$'s gradient on the value of other variables) between different coordinates. 
In a typical iteration, they aim to move all variables simultaneously proportional to the respective gradients, without changing the gradient of any variable by too much. 
When the gradients of the variables are not on the same scale, a natural obstacle arises: when variable $i$ has a large gradient, we would like to move $i$ by a large step in order to harness the gradient improvement, but we are prevented from doing so, due to the interference of $i$ with some other variable $j$, which has a tiny gradient. 

We tackle this bottleneck by designing in a dynamic manner selective coordinate descent steps designed to reduce interference. 
In each iteration, we group all the variables according to the magnitudes of their gradients such that variables in the same group all have approximately the same magnitudes up to some constant factor, as specified in~\eqref{eq:groups} in Section~\ref{sxn:main-results-packing}. 
If we then only update variables of one randomly chosen group, as in Step~\ref{step:selective-update} of Algorithm~\ref{alg:PP}, then we can take larger steps for the subset of variables we update. 
In addition, we show that we only need $\log(1/\epsilon)$ groups, so each iteration we update a large fraction of the variables on expectation. 

To make our analysis work out, we follow~\cite{AO15-parallel}, and we interpret the update step as both a gradient descent step and a mirror descent step.
(We note, though, that this is different from~\cite{AO15-stochastic,AO14}, where the gradient descent step and the mirror descent step are separate steps, and the algorithm takes a linear coupling of the two steps to achieve acceleration as in Nesterov's accelerated gradient descent. In~\cite{AO15-parallel} and in our algorithm, it is a single update step each iteration, interpreted as both a gradient step and a mirror descent step.) 
This two-way interpretation is required when we use convexity to upper bound the gap between $f_\mu(x_k)$ and $f_\mu(u)$ for some $x_k,u$. Recall we break the gradients into small, medium and large components as in~\eqref{eq:smallLarge}, so we have:
\[
f_\mu(x_k)-f_\mu(u)\leq \dotp{\nabla f_\mu(x_k),x_k-u}=\dotp{\zeta_k,x_k-u}+\dotp{\xi_k,x_k-u}+\dotp{\eta_k,x_k-u}   .
\]
We will see that in expectation, the loss incurred by the medium component will be bounded using the mirror descent interpretation in Lemma~\ref{lemma:mirror}, and the loss incurred by the large component will be bounded using the gradient descent interpretation in Lemma~\ref{lemma:gradient}.

A benefit of our DB-SCD method is that it is modular and thus it can be coupled cleanly with existing methods. 
To illustrate this, since our interest in this particular problem was inspired by the original improvement of~\cite{AO15-parallel}, we will to the extent possible adopt the techniques from their work, pointing out similarities in the following sections whenever possible.

\section{Analysis of Algorithm for Packing LPs: Proof of Theorem~\ref{thm:main}}
\label{sxn:analysis}


In this section, we will provide a proof of Theorem~\ref{thm:main}.
Recall that we denote by $\opt$ the optimal value of~\eqref{eq:packLP} and that Algorithm~\ref{alg:PP} will compute a {\it $(1-\epsilon)$-approximation} $x$ where $Ax\leq \vone$ and $\vone^Tx\geq (1-\epsilon)\opt$.
In Section~\ref{sxn:solver-prelim}, we'll present some preliminaries and describe how we perform smoothing on the original packing objective function. 
We'll analyze the update step as a gradient descent step in Section~\ref{sxn:solver-gradient}, and we'll analyze the same update step as a mirror descent step in Section~\ref{sxn:solver-mirror}. 
Finally, in Section~\ref{sxn:solver-coupling}, we'll show how to combine the two analyses to complete the proof of Theorem~\ref{thm:main}. 
Some of the following results are technically-tedious but conceptually-straightforward extensions of analogous results from~\cite{AO15-parallel}, and some of the results are restated from~\cite{AO15-parallel}.
For completeness, we provide the proof of all of these results, with the latter relegated to Appendix~\ref{App:proofs}.

\subsection{Preliminaries and Smoothing the Objective}
\label{sxn:solver-prelim}

To start, let's assume  that
\[
\min_{i\in [n]}\maxAi=1   .
\]
This assumption is without loss of generality: since we are interested in multiplicative $(1-\epsilon)$-approximation, we can simply scale $A$ for this to hold without sacrificing approximation quality.
With this assumption, the following lemma holds.
(This lemma is the same as Proposition $2.2.(a)$ in~\cite{AO15-parallel}, and its proof is included for completeness in Appendix~\ref{App:proofs}.)
\begin{restatable}{lemma}{optrange}
$\opt\in [1,n]$
\end{restatable}

\noindent
With $\opt$ being at least $1$, the error we introduce later in the smoothing step will be small enough that the smoothing function approximates the packing LP well enough with respect to $\epsilon$ around the optimum. 


We will turn the packing LP objective into a smoothed objective function $f_\mu(x)$, as used in~\cite{AO15-parallel,AO15-stochastic}, and we are going to find a $(1-\epsilon)$-approximation of the packing LP by approximately minimizing $f_\mu(x)$ over the region $x\geq 0$.
The function $f_\mu(x)$ is 
\[
f_\mu(x)\myeq -\vone^Tx+\max_{y\geq 0}\{y^T(Ax-\vone)+\mu H(y)\}   ,
\]
and it is a smoothed objective in the sense that it turns the packing constraints into soft penalties, with $H(y)$ being a regularization term. 
Here, we use the generalized entropy $H(y)=-\sum_j y_j\log y_j + y_j$, where $\mu$ is the smoothing parameter balancing the penalty and the regularization. 
It is straightforward to compute the optimal $y$, and write $f_\mu(x)$ explicitly, as stated in the following lemma.

\begin{lemma}
\label{lemma:explicitFunction}
$f_\mu(x)=-\vone^Tx+\mu \sum_{j=1}^m p_j(x)$, where $p_j(x)\myeq \exp(\frac{1}{\mu}((Ax)_j)-1)$.
\end{lemma}

Optimizing $f_\mu(x)$ gives a good approximation to $\opt$, in the following sense. 
If we let $x^*$ be an optimal solution, and $u^*\myeq (1-\epsilon/2)x^*$, then we have the properties in the following lemma.
(This lemma is the same as Proposition $2.2$ in~\cite{AO15-parallel}, and its proof is included for completeness in Appendix~\ref{App:proofs}.)
\begin{restatable}{lemma}{smoothing}\label{lemma:smoothing_properties}
Setting the smoothing parameter $\mu=\frac{\epsilon}{4\log(nm/\epsilon)}$, we have
\begin{enumerate}
\item $f_\mu(u^*)\leq -(1-\epsilon)\opt$.
\item $f_\mu(x)\geq -(1+\epsilon)\opt$ for every $x\geq 0$.
\item Letting $x_0\geq 0$ be such that $x_0[i]=\frac{1-\epsilon/2}{n\maxAi}$ for each $i\in [n]$, we have $f_\mu(x_0)\leq -\frac{1-\epsilon}{n}$.
\item For any $x\geq 0$ satisfying $f_\mu(x)\leq 0$, we must have $Ax\leq (1+\epsilon)\vone$, and thus $\vone^Tx\leq (1+\epsilon)\opt$.
\item If $x\geq 0$ satisfies $f_\mu(x)\leq -(1-O(\epsilon))\opt$, then $\frac{1}{1+\epsilon}x$ is a $(1-O(\epsilon))$-approximation to the packing~LP.
\item The gradient of $f_\mu(x)$ is
\[
\nabla f_\mu(x)=-\vone + A^T\overrightarrow{p(x)} \quad \text{where} \quad p_j(x)\myeq \exp(\frac{1}{\mu}((Ax)_j-1) ,
\]
and $\nabla_i f_\mu(x)=-1+\sum_j A_{ji}p_j(x)\in [-1,\infty]$.
\end{enumerate}
\end{restatable}

\noindent
Although $f_\mu(x)$ gives a good approximation to the packing LP without introducing dependence of $\opt$, we cannot simply apply the standard (accelerated) gradient descent algorithm to optimize it, as $f_\mu(x)$ doesn't have the necessary Lipschitz-smoothness property. 
(Indeed, our DB-SCD method was designed to address this issue.)

As before~\cite{WRM15}, we interpret our update step $x_{k+1}[i]\leftarrow \newx[i]\myeq x_k[i]\exp(-\alpha \xi^{(t)}_k[i])$ in Algorithm~\ref{alg:PP} as both a gradient descent step as well as a mirror descent step. 
That is, in order to prove the theorem, our analysis will view it from both perspectives.
We proceed now with the respective analysis for the two interpretations.

\subsection{Gradient Descent Step}
\label{sxn:solver-gradient}

We will first analyze the update step in Algorithm~\ref{alg:PP} as a gradient descent step. 
As in most gradient descent analysis, we need to bound our step's impact on the gradients. 
To do so, we will show $f_\mu(x)$ is \emph{locally multiplicative Lipschitz continuous}, in a sense quantified by the following lemma. 
Note the result is a stronger version of Proposition $3.6$ in~\cite{AO15-parallel}, in the sense that the step size $\alpha$ is $1/\epsilon$ larger. 
This improvement is achieved due to the reduced interference from our DB-SCD updating method.

\begin{lemma}
\label{lemma:Lipschitz}
Let $\newx[i]=x_k[i]\exp(-\alpha \xi^{(t)}_k[i])$, for any $t=0,\ldots,w-1$, as in Algorithm~\ref{alg:PP}. 
Let $B_t=\{i|\xi^{(t)}_k[i]>0\}$ be the set of variables we update. 
If $Ax_k\leq (1+\epsilon)\vone$, then for any $x=\tau x_k+(1-\tau)\newx$ where $\tau\in [0,1]$, we have $\forall i\in B_t, \deli{x}$ is between $\frac{1}{2}\deli{x_k}$ and $\frac{3}{2}\deli{x_k}$.
\end{lemma}
\begin{proof}
Because for all $i\in B_t$, $\xi^{(t)}_k[i] \in (\epsilon 2^t,\epsilon 2^{t+1}]\cup [-\epsilon 2^{t+1},-\epsilon 2^t)$, each variable changes multiplicatively by at most $\exp(\pm \alpha \epsilon 2^{t+1})$, and since $\alpha \epsilon 2^{t+1} \leq 1/4$, we must have for all $i$,
\begin{equation}
\label{eq:xrange}
x[i]\in x_k[i]\cdot [1-\frac{8}{3}\alpha\epsilon 2^t,1+\frac{8}{3}\alpha\epsilon 2^t]  .
\end{equation}
Now we look at the impact of the step on the exponential penalties
\[
p_j(x)=\exp(\frac{1}{\mu}((Ax)_j-1))  .
\]
Due to~\eqref{eq:xrange}, and $(Ax_k)_j\leq (1+\epsilon)$ for any $j$, we have
\[
|(Ax)_j-(Ax_k)_j|\leq \frac{8}{3}\alpha\epsilon 2^t (Ax_k)_j\leq \frac{10}{3}\alpha\epsilon2^t  .
\]
Then by our choice of $\alpha$, we have
\[
p_j(x)\geq p_j(x_k)\exp (-\frac{10\alpha\epsilon 2^t}{3\mu})=p_j(x_k)\exp(-\frac{\epsilon 2^t}{6})  .
\]
Since $\epsilon 2^t\leq 1$, we have $\exp(-\frac{\epsilon 2^t}{6})\geq (1-\frac{\epsilon 2^t}{4})$. By a similar calculation for the upper bound, we have
\begin{equation}
\label{eq:prange}
p_j(x)\in p_j(x_k)\cdot[1-\frac{\epsilon 2^t}{4},1+\frac{\epsilon 2^t}{4}]  .
\end{equation}
For any $i\in B_t$, if $\xi^{(t)}_k[i]\in (\epsilon 2^t,\epsilon 2^{t+1}]$, we have
\begin{align*}
\nabla_i f_\mu(x) &=(A^Tp(x))_i-1\\
&> (A^Tp(x_k))(1-\frac{\epsilon 2^t}{4})-1\\
&= (\nabla_i f_\mu(x_k) +1)(1-\frac{\epsilon 2^t}{4})-1\\
&\geq \frac{\nabla_i f_\mu(x_k)}{2}  ,
\end{align*}
where the last step is due to $\epsilon 2^t\leq 1$ and $\nabla_i f_\mu(x_k)\geq \xi^{(t)}_k[i]>\epsilon 2^t$. 
By similar calculation, we get $\nabla_i f_\mu(x)\leq \frac{3}{2}\nabla_i f_\mu(x_k)$. The same holds for the case $\xi^{(t)}_k[i]\in [-\epsilon 2^{t+1},-\epsilon 2^t)$.
\end{proof}

We will see in Claim~\ref{claim:approxFeasible} that the condition of $Ax_k\leq (1+\epsilon)\vone$ holds for all $k=0,\ldots,T$. 
Once we establish smoothness of the gradients within the range of our update step, we can lower bound the improvement we make. 
In particular, the term $\dotp{\alpha\eta_k^{(t)},x_k-u}$ is the loss incurred from the truncation, as our update step doesn't act on the truncated part, but it shows up when we use convexity to bound the gap to optimality.
\begin{lemma}
\label{lemma:gradient}
For all $t=0,\ldots,w-1$, any $u\geq 0$
\[
\dotp{\alpha\eta^{(t)}_k,x_k-u}\leq 4(f_\mu(x_k)-f_\mu(\newx))   .
\]
\end{lemma}
\begin{proof}
First observe that
\begin{align*}
f_\mu(x_k)-f_\mu(\newx)&=\int_0^1\dotp{\nabla f_\mu(\newx+\tau(x_k-\newx)),x_k-\newx} d\tau\\
&=\sum_{i\in B_t}\int_0^1\nabla_i f_\mu(\newx+\tau(x_k-\newx)) d\tau\times (x_k[i]-\newx[i])  ,
\end{align*}
where the last equality is because $x_k[i]-\newx[i]=0$ for $i\not\in B_t$. 
By Lemma~\ref{lemma:Lipschitz}, we have that $\deli{\newx+\tau(x_k-\newx)}$ has the same sign as $\deli{x_k}$ for all $\tau\in[0,1]$. 
Furthermore, by our update rule, $x_k[i]-\newx[i]$ also has the same sign as $\deli{x_k}$, and so we have $f_\mu(x_k)-f_\mu(\newx)\geq 0$ for all $t$. 
If $t<w-1$, then we know $\eta^{(t)}_k=\vec{0}$, and thus
\[
\dotp{\alpha\eta^{(t)}_k,x_k-u}=0\leq 4(f_\mu(x_k)-f_\mu(\newx))  .
\]
When $t=w-1$, let $B=\{i|\deli{x_k}>1\}\supseteq B_t$ be the set of variables with nonzero $\eta^{(t)}_k[i]$, we know for $i\in B$, $\xi^{(t)}_k[i]=1$, so $\newx[i]=x_k[i]\exp(-\alpha)$, and
\begin{align*}
f_\mu(x_k)-f_\mu(\newx)&=\sum_{i\in B_t}\int_0^1\nabla_i f_\mu(\newx+\tau(x_k-\newx)) d\tau\times (x_k[i]-\newx[i])\\
&\geq \sum_{i\in B}\int_0^1\nabla_i f_\mu(\newx+\tau(x_k-\newx)) d\tau\times (x_k[i]-\newx[i])\\
&\geq \sum_{i\in B}\frac{1}{2}\deli{x_k} \times x_k[i](1-\exp(-\alpha))\\
&\geq \sum_{i\in B}\frac{\alpha}{4}\deli{x_k} x_k[i]   .
\end{align*}
The first inequality is due to $B_t\subseteq B$, and every $i$ has non-negative contribution to the sum. The second inequality is from Lemma~\ref{lemma:Lipschitz}, and the last inequality is because $(1-\exp(-\alpha))>\alpha/2$ when $\alpha<1/10$. 
Then we have
\[
\dotp{\alpha\eta^{(t)}_k,x_k-u}\leq \sum_{i\in B}\alpha\deli{x_k} x_k[i]\leq 4(f_\mu(x_k)-f_\mu(\newx))  ,
\]
where the first inequality is because $\deli{x_k}>\eta^{(t)}_k\geq 0$ in this case, and $u\geq 0$.
\end{proof}

We see $f_\mu(x_k)-f_\mu(\newx)\geq 0$ for any $t=0,\ldots,w-1$, we have the following.

\begin{claim}\label{claim:approxFeasible}
$f_\mu(x_k)$ is non-increasing with $k$. 
By part $(3),(4)$ of Lemma~\ref{lemma:smoothing_properties}, $Ax_k\leq (1+\epsilon)\vone$, and $\vone^Tx_k\leq (1+\epsilon)\opt$ for all $k$.
\end{claim}

\subsection{Mirror Descent Step}
\label{sxn:solver-mirror}

We now interpret the update step as a mirror descent step. We use the same proximal setup as in~\cite{AO15-parallel}. The distance generating function will be the generalized entropy function, where
\[
w(x)\myeq \sum_{i\in[n]}x[i]\log(x[i])-x[i]   ,
\]
and the corresponding Bregman divergence function will be
\[
V_x(y)=\sum_{i\in[n]}(y[i]\log\frac{y[i]}{x[i]}+x[i]-y[i])   .
\]
This is the standard proximal setup when one works with $L_1$-norm with the simplex as the feasible region. 
In our case, since the feasible region is $x\geq 0$, we don't have the standard strong convexity of the Bregman divergence, but one can verify
\begin{equation}\label{eq:Bregman}
V_x(y)=\sum_{i\in[n]}(y[i]\log\frac{y[i]}{x[i]}+x[i]-y[i])\geq \sum_{i\in[n]}\frac{(x[i]-y[i])^2}{2\max\{x[i],y[i]\}} .
\end{equation}
To interpret the update step as a mirror descent step, the following claim is used.  It is the same as Claim $3.7$ in~\cite{AO15-parallel} applied to different vectors. 
It is fairly straightforward to verity, and we include the proof in Appendix~\ref{App:proofs}.

\begin{restatable}{claim}{mirrorstep}
For all $t=0,\ldots,w-1$, we have
\[
\newx=\argmin_{z\geq 0}\{V_{x_k}(z)+\dotp{z-x_k,\alpha\xi^{(t)}_k}\}  .
\]
\end{restatable}

Once we see the update step is indeed a mirror descent step, we can derive the following result from the textbook mirror descent analysis (or, e.g., Lemma $3.3$ in~\cite{AO15-parallel}).

\begin{lemma}\label{lemma:mirror}
For all $t=0,\ldots,w-1$, we have for any $u\geq 0$
\[
\dotp{\alpha\xi^{(t)}_k,x_{k}-u}\leq \alpha^2\opt+V_{x_{k}}(u)-V_{\newx}(u)   .
\]
\end{lemma}
\begin{proof}
The lemma follows from the following chain of equalities and inequalities.
\begin{align*}
\dotp{\alpha\xi^{(t)}_k,x_{k}-u}&=\dotp{\alpha\xi^{(t)}_k,x_{k}-\newx}+\dotp{\alpha\xi^{(t)}_k,\newx-u}\\
&\leq \dotp{\alpha\xi^{(t)}_k,x_{k}-\newx}+\dotp{-\nabla V_{x_k}(\newx),\newx-u}\\
&\leq \dotp{\alpha\xi^{(t)}_k,x_{k}-\newx}+V_{x_{k}}(u)-V_{\newx}(u)-V_{x_{k}}(\newx) \\
&\leq \sum_{i\in[n]}\left(\alpha\xi^{(t)}_k[i](x_{k}[i]-\newx[i])-\frac{(x_k[i]-\newx[i])^2}{2\max\{x_k[i],\newx[i]\}}\right) +V_{x_{k}}(u)-V_{\newx}(u) \\
&\leq \sum_{i\in[n]}\frac{(\alpha\xi^{(t)}_k[i])^2\max\{x_k[i],\newx[i]\}}{2}+V_{x_{k}}(u)-V_{\newx}(u) \\
&\leq \frac{2}{3}\alpha^2\vone^Tx_k+V_{x_{k}}(u)-V_{\newx}(u)\\
&\leq \alpha^2\opt+V_{x_{k}}(u)-V_{\newx}(u)   .
\end{align*}
The first equality follows by adding and subtracting $\newx$. 
The first inequality is due to the the minimality of $\newx$, which gives
\[
\dotp{\nabla V_{x_{k}}(\newx)+\alpha\xi^{(t)}_k,u-\newx}\geq 0 \quad \forall u\geq 0   ,
\]
the second inequality is due to the standard three point property of Bregman divergence, that is $\forall x,y\geq 0$
\[
\dotp{-\nabla V_x(y),y-u}=V_x(u)-V_y(u)-V_x(y)   ,
\]
the third inequality is from~\eqref{eq:Bregman}, the fourth inequality follows from $2ab-a^2\leq b^2$, the next inequality is due to $\newx[i]\leq x_k[i](1+\epsilon)$, and $\xi_k^{(t)}[i]\leq 1$. 
The last inequality is by Claim~\ref{claim:approxFeasible}, $\vone^Tx_k\leq (1+\epsilon)\opt$.
\end{proof}

\subsection{Coupling of Gradient and Mirror Descent}
\label{sxn:solver-coupling}

In this section we show convergence using the results we derived by analyzing the update step as both a gradient descent step and a mirror descent step. 

Recall we break the gradients into small, medium and large components. The proof follows a similar approach as Lemma $3.4$ of~\cite{AO15-parallel}, where we bound the three components respectively, and telescope along all iterations. Furthermore, we divide the medium and large components into $w=\log(\ceil*{\frac{1}{\epsilon}})$ groups, as follows:
\[
\nabla f_\mu(x_k)=\zeta_k+\xi_k+\eta_k
\]
and
\[
\xi_k=w\E_t[\xi^{(t)}_k]  \quad\mbox{and}\quad   \eta_k=w\E_t[\eta^{(t)}_k]   .
\]
We bound the gap to optimality at iteration $k$, as follows:
\begin{align*}
\alpha(f_\mu(x_k)-f_\mu(u^*))
\leq& \dotp{\alpha\nabla f_\mu(x_k),x_k-u^*}\\
=&\alpha\dotp{\zeta_k,x_k-u^*}+\alpha\dotp{\xi_k,x_k-u^*}+\alpha\dotp{\eta_k,x_k-u^*}\\
=&\alpha\dotp{\zeta_k,x_k-u^*}+w\E_t[\dotp{\alpha\xi^{(t)}_k,x_k-u^*}]+w\E_t[\dotp{\alpha\eta^{(t)}_k,x_k-u^*}]   .
\end{align*}
The first line is due to convexity. 
The next two lines just break and regroup the terms. Now we upperbound each of the three terms
\begin{lemma}\label{lemma:componentBound}
\begin{enumerate}
We have the following.
\item $\dotp{\zeta_k,x_k-u^*}\leq 3\epsilon\opt$.
\item $\forall t$, $\dotp{\alpha\xi^{(t)}_k,x_k-u^*}\leq \alpha^2\opt+V_{x_k}(u^*)-V_{\newx}(u^*)$.
\item $\forall t$, $\dotp{\alpha\eta^{(t)}_k,x_k-u^*}\leq 4(f_\mu(x_k)-f_\mu(\newx))$.
\end{enumerate}
\end{lemma}
\begin{proof}
We establish each result in turn.
\begin{enumerate}
\item We know $|\zeta_k[i]|\leq \epsilon$ for all $i$, $\vone^Tx_k\leq (1+\epsilon)\opt$ from Claim~\ref{claim:approxFeasible}, and $\vone^Tu^*\leq \opt$, thus
\[
\dotp{\zeta_k,x_k-u^*}\leq \epsilon(\vone^Tx_k+\vone^Tu^*)\leq 3\epsilon\opt
\]
\item This is just Lemma~\ref{lemma:mirror} applied to $u=u^*$.
\item This is just Lemma~\ref{lemma:gradient} applied to $u=u^*$.
\end{enumerate}
\end{proof}

\noindent
Given this, we have
\begin{align*}
\alpha(f_\mu(x_k)-f_\mu(u^*)) \leq& \alpha\dotp{\zeta_k,x_k-u^*}+w\E_t[\dotp{\alpha\xi^{(t)}_k,x_k-u^*}]+w\E_t[\dotp{\alpha\eta^{(t)}_k,x_k-u^*}] \nonumber \\
\leq & 3\alpha\epsilon\opt+w\E_t[\alpha^2\opt+V_{x_k}(u^*)-V_{\newx}(u^*)]+4wf_\mu(x_k)-4w\E_t[f_\mu(\newx)]   .
\end{align*}
The above inequality holds for $x_k$ following any sequence of random choices (i.e., $t$'s in the first $k-1$ iterations). 
Let $I_k$ denote the random choices over the first $k$ iterations, and we take expectation of the above inequality to~get
\begin{align}
\label{eq:iterationGap}
\alpha(\E_{I_k}[f_\mu(x_k)]-f_\mu(u^*)) \leq & 3\alpha\epsilon\opt+\alpha^2w\opt+ w\E_{I_{k}}[V_{x_k}(u^*)]-w\E_{I_{k+1}}[V_{\newx}(u^*)] \nonumber\\
& +4w\E_{I_k}[f_\mu(x_k)]-4w\E_{I_{k+1}}[f_\mu(\newx)]    .
\end{align}
Telescoping~\eqref{eq:iterationGap} for $k=0,\ldots,T-1$, we get
\begin{align}
\alpha \sum_{k=0}^{T-1}(\E_{I_k}[f_\mu(x_k)]-f_\mu(u^*))\leq& 3T\alpha\epsilon\opt + wT\alpha^2\opt +wV_{x_0}(u^*)+4wf_\mu(x_0)-4w\E_{I_T}[f_\mu(x^{(t)}_T)] \nonumber\\
\label{eq:telescope}
\leq& 3T\alpha\epsilon\opt + wT\alpha^2\opt +2w\log 2n\opt-4w\E_{I_T}[f_\mu(x^{(t)}_T)]   ,
\end{align}
where the last inequality is due to $f_\mu(x_0)<0$, and 
\begin{claim}
$V_{x_0}(u^*)\leq 2\log 2n\opt$
\end{claim}
\begin{proof}
\begin{align*}
V_{x_0}(u^*)=\sum_i u^*[i]\log\frac{u^*[i]}{x_0[i]}+x_0[i]-u^*[i] &\leq\sum_i u^*[i]\log\frac{u^*[i]}{x_0[i]}+x_0[i]\\
&\leq \sum_i u^*[i]\log\frac{1/\maxAi}{(1-\epsilon/2)/(n\maxAi)}+\frac{1-\epsilon/2}{n\maxAi}\\
&\leq \vone^Tu^*\log(2n)+1\leq 2\log(2n)\opt
\end{align*}
where we have used $u^*_i\leq \frac{1}{\maxAi}$ in the second line, since $Au^*\leq \vone$. The third line is due to $\vone^Tu^*\leq \opt$, and $\opt\geq 1$.
\end{proof}

\noindent
We prove that $E_{I_T}[f_\mu(x^{(t)}_T)]\leq -(1-5\epsilon)\opt$ (i.e. Theorem~\ref{thm:main}) by contradiction. 
If $E_{I_T}[f_\mu(x^{(t)}_T)]> -(1-5\epsilon)\opt$, we have $-4wE_{I_T}[f_\mu(x^{(t)}_T)]\leq 4w\opt$. If we divide both sides of~\eqref{eq:telescope} by $\alpha T$, then we have
\begin{align*}
&
\hspace{-20mm}
\frac{1}{T\alpha}\sum_{k=0}^{T-1}\alpha(\E_{I_k}[f_\mu(x_k)]-f_\mu(u^*))\\
\leq& \frac{1}{T\alpha}(3T\alpha\epsilon\opt + wT\alpha^2\opt +2w\log 2n\opt-4w\E_{I_T}[f_\mu(x^{(t)}_T)])\\
\leq& 3\epsilon\opt+w\alpha\opt+\frac{2w\log 2n}{T\alpha}\opt+\frac{4w}{T\alpha}\opt   .
\end{align*}
Recall $\alpha=\mu/20=\frac{\epsilon}{20\log \frac{mn}{\epsilon}}$, we have $w\alpha\leq \epsilon/20$. 
By our choice of $T=\frac{10w\log 2n}{\alpha\epsilon}$, we have $\frac{2w\log 2n}{T\alpha}<\epsilon$, and $\frac{4w}{T\alpha}<\epsilon$. 
Thus
\[
\frac{1}{T}\sum_{k=0}^{T-1}(\E_{I_k}[f_\mu(x_k)]-f_\mu(u^*))\leq 4\epsilon\opt
\]
From part $(1)$ of Lemma~\ref{lemma:smoothing_properties}, we know $f_\mu(u^*)\leq -(1-\epsilon)\opt$, which suggests there exists a $x_k$, such that $\E_{I_T}[f_\mu(x_k)]\leq -1(1-5\epsilon)\opt$. 
This gives a contradiction of $\E_{I_T}[f_\mu(x^{(t)}_T)]> -(1-5\epsilon)\opt$ by Claim~\ref{claim:approxFeasible}, as $f_\mu(x_k)$ is non-increasing, so is $\E_{I_k}[f_\mu(x_k)]$.

The running time guarantee in Theorem~\ref{thm:main} comes directly from our choice of $T$, and that in each iteration of Algorithm~\ref{alg:PP}, the gradients can be computed in $O(\log N)$ distributed iterations and $O(N)$ total work.

\section{Analysis of Algorithm for Covering LPs: Proof of Theorem~\ref{thm:dual}}
\label{sxn:proof-dual-thm}

In this section, we will provide a proof of Theorem~\ref{thm:dual}.
To do so, first recall some properties from the analysis of the packing algorithm:
\begin{enumerate}
\item $\forall u\geq 0$,
\begin{equation}\label{eq:covSolCondition}
\begin{array}{r@{}l}
\frac{1}{T}\E[\sum_k f_\mu(x_k)]-f_\mu(u) &{}\leq \frac{1}{T}\E[\sum_{k=0}^{T-1}\dotp{\nabla f_\mu(x_k),x_k-u}]\\
 &{} \leq \frac{4w}{\alpha T}(f_\mu(x_0)-E[f_\mu(x_T)])+\frac{w}{\alpha T}V_{x_0}(u)+2\epsilon\opt + \epsilon \vone^T u  .
\end{array}
\end{equation}
This is simply telescoping~\eqref{eq:iterationGap} for a general $u\geq 0$ instead of $u^*$. Notice Lemma~\ref{lemma:componentBound}(1) for general $u\geq 0$ gives $\dotp{\zeta_k,x_k-u}\leq \epsilon\vone^Tu+\epsilon \vone^Tx_k\leq 2\epsilon\opt + \vone^Tu$.
\item $Ax_k\leq (1+\epsilon)\vone, \vone^Tx_k\leq (1+\epsilon)\opt$ and $f_\mu(x_k)\geq -(1+\epsilon)\opt$ hold for all $k$ and any outcome of random choices. This follows from Claim~\ref{claim:approxFeasible}, and Lemma~\ref{lemma:smoothing_properties}(2). Also $x_0[i]=\frac{1-\epsilon/2}{n\maxAi}$ for each $i\in [n]$, and $f_\mu(x_0)\leq 0$.
\end{enumerate}

\noindent
We start with the following lemma which states that $\vone^T\bar{y}$ is close to $\opt$ on expectation.

\begin{lemma}
\label{lemma:bary}
For any $T\geq \frac{6w}{\alpha\epsilon}$, we have $\E[\vone^T\bar{y}]\leq (1+5\epsilon)\opt$
\end{lemma}

\noindent
The proof of this lemma follows directly from Lemma $D.1$ of~\cite{AO15-parallel}, only with the additional $w=\log(\ceil*{\frac{1}{\epsilon}})$ due to our dynamic grouping, and the expectation. 
The expectation holds since all the inequalities used in the proof hold universally (i.e., in any outcome of the random choices). 
We omit the detailed proof here, and encourage interested readers to look at~\cite{AO15-parallel}.

Now we look at the $i$-th constraint of the covering LP, which corresponds to the variable $x[i]$ in the dual packing instance. 
Let $Z_k^{(i)}$ be the indicator random variable of whether $x[i]$ is in the group being updated in iteration $k$ of Algorithm~\ref{alg:PP}, and let
\[
S_i=w\sum_{k=0}^{T-1}Z_k^{(i)}(\min\{\deli{x_k},1\}+\epsilon)  .
\]
We can obtain a lower bound on the random variable $S_i$ as follows:
\begin{lemma}
\begin{equation}
\label{eq:detBound}
S_i\geq -\frac{w\log 2n}{\alpha}   \qquad \forall i.
\end{equation}
\end{lemma}
\begin{proof}
Using the notations of Algorithm~\ref{alg:PP}, we have, for all $i$, $\frac{S_i}{w}\geq \sum_{k=0}^{T-1}\xi_k^{(t)}[i]$. 
We know the cumulative update on variable $x[i]$ must be bounded, due to Claim~\ref{claim:approxFeasible}, $x_k[i]\leq \frac{1+\epsilon}{\maxAi}$ for all $k$, and in particular 
\[
\frac{1+\epsilon}{\maxAi}\geq x_T[i]=x_0[i]\exp(-\alpha \sum_{k=0}^{T-1}\xi_k^{(t)}[i])\geq x_0[i]\exp(-\alpha \frac{S_i}{w})=\frac{1-\epsilon/2}{n\maxAi}\exp(-\alpha \frac{S_i}{w})   .
\]
The bound in~\eqref{eq:detBound} follows.
\end{proof}
Notice that the slackness of the $i$-th covering constraint with the solution $\bar{y}$ is
\begin{align*}
(A^T\bar{y})_i-1+\epsilon=&\frac{1}{T}\sum_{k=0}^{T-1}(A^T\overrightarrow{p(x_k)})_i-1+\epsilon\\
=&\frac{1}{T}\sum_{k=0}^{T-1}\deli{x_k}+\epsilon\\
\geq & \frac{1}{T}\sum_{k=0}^{T-1}\min\{\deli{x_k},1\}+\epsilon\\
\myeq & \frac{1}{T} U_i   ,
\end{align*}
with the definition of the random variable $U_i=\sum_{k=0}^{T-1}\min\{\deli{x_k},1\}+\epsilon$. 
If the $i$-th variable is updated in all iterations, i.e., $Z_k^{(i)}=1$ for all $k$, then we have $U_i=S_i$ in that case, and when $T\geq \frac{6w\log 2n}{\alpha \epsilon}$, we have that
\[
(A^T\bar{y})_i-1+\epsilon\geq \frac{1}{T}S_i\geq -\frac{w\log 2n}{\alpha T}\geq -\epsilon   .
\]
Thus, we know that $(A^T\bar{y})_i\geq 1-2\epsilon$, for all $i$, which means all covering constraints are approximately feasible. 
However, we don't always update variable $i$ in all iterations of Algorithm~\ref{alg:PP}, and so we need to bound the difference $S_i-U_i$. 
We do so with the following lemma.

\begin{lemma}
\label{lemma:martingale}
For any $T\geq \frac{2w^2\log\frac{n}{\epsilon} }{\epsilon^2}$, we have $\Pr[S_i-U_i\geq \epsilon T]\leq \frac{\epsilon}{n}$.
\end{lemma}
\begin{proof}
The randomness of $U_i$ and $S_i$ comes from the random choice of which group to update in each iteration of Algorithm~\ref{alg:PP}. 
Let 
\[
D_k^{(i)}=w\sum_{k'\leq k}Z_{k'}^{(i)}(\min\{\deli{x_k},1\}+\epsilon)-\sum_{k'\leq k}\min\{\deli{x_k},1\}+\epsilon  .
\]
Let $G_k$ be the random choice (i.e., the group to update) made at $k$-th iteration of Algorithm~\ref{alg:PP}. 
Since $Z_{k}^{(i)}$ is an indicator random variables with probability of $\frac{1}{w}$ being $1$, and it is independent from $G_0,\ldots,G_{k-1}$, $D_k^{(i)}$ is a martingale with respect to $G_k$, as
\[
E[D_k^{(i)}|G_0,\ldots,G_{k-1}]=D_{k-1}^{(i)}+\frac{1}{w}w(\min\{\deli{x_k},1\}+\epsilon)-(\min\{\deli{x_k},1\}+\epsilon)=D_{k-1}^{(i)}  ,
\]
 and so we have
\[
D_0^{(i)}=E[S_i-U_i]=0,\quad D_T^{(i)}=S_i-U_i    .
\]
Furthermore, $|D_k^{(i)}-D_{k+1}^{(i)}|\leq w$ for all $k$, so we can apply Azuma's inequality, and get
\[
\Pr[S_i-U_i\geq \epsilon T]\leq \exp(\frac{\epsilon^2T^2}{2Tw^2})\leq \frac{\epsilon}{n}  ,
\]
from which the lemma follows.
\end{proof}
The above lemma, Lemma~\ref{lemma:martingale}, shows that with high probability, $\bar{y}$ satisfies the $i$-th covering constraint up to $-3\epsilon$. 
In the rare case it doesn't, we use Algorithm~\ref{alg:fix} to fix it, and get $\bar{y}'$. We show on expectation this step doesn't add too much to the total cost.
\begin{lemma}
$\E[\vone^T\bar{y'}]\leq (1+6\epsilon)\opt$.
\end{lemma}
\begin{proof}
When $S_i-U_i\leq \epsilon T$, we have 
\[
(A^T\bar{y})_i-1+\epsilon\geq \frac{1}{T}(S_i-\epsilon T)\geq -\frac{w\log 2n}{\alpha T}-\epsilon\geq -2\epsilon  ,
\]
and so we don't need to fix the $i$-th constraint. 
When that is not the case, since $(A^T\bar{y})_i\geq 0$, and $\maxAi\geq 1$, we need to add at most $1$ to some variable $\bar{y}'_j$ to fix the $i$-th covering constraint. 
For all the $n$ covering constraints, we add on expectation at most $n\frac{\epsilon}{n}\leq \epsilon\leq \epsilon\opt$ to $\bar{y}$ to get $\bar{y}'$. 
Together with Lemma~\ref{lemma:bary}, we have $\E[\vone^T\bar{y'}]\leq (1+6\epsilon)\opt$.
\end{proof}
We complete the proof of Theorem~\ref{thm:dual} by noticing $(A^T\bar{y}')_i\geq 1-3\epsilon$, for all $i$.
Thus, the output of Algorithm~\ref{alg:fix}, $\frac{\bar{y}'}{1-3\epsilon}$, satisfies the properties stated in Theorem~\ref{thm:dual}.

\vspace{0.4in}
\noindent
\textbf{Acknowledgments.}
DW was supported by ARO Grant W911NF-12-1-0541 and NSF Grant CCF-1528174, SR was funded by NSF Grant CCF-1528174, and MWM acknowledges the support of the NSF, AFOSR, and DARPA.
\vspace{0.1in}

\begin{appendices}

\section{Missing Proofs}
\label{App:proofs}
The following proofs can be found in~\cite{AO15-parallel}, and we include them here for completeness.
\optrange*
\begin{proof}
By the assumption $\min_{i\in [n]}\maxAi=1$, we know at least one variable has all coefficients at most $1$, so we can just set that variable to $1$, which gives $\opt\geq 1$. On the other hand, since each variable has a coefficient of $1$ in some constraint, no variable can be larger than $1$, thus $\opt\leq n$.
\end{proof}
\smoothing*
\begin{proof}
We establish each result in turn.
\begin{enumerate}
\item Since $Ax^*\leq \vone$, and $u^*=(1-\epsilon/2)x^*$, we have $(Au^*)_j-1\leq -\epsilon/2$ for all $j$. Then $p_j(u^*)\leq \exp(-\frac{1}{\mu}\frac{\epsilon}{2})=(\frac{\epsilon}{mn})^2$, and $f_\mu(u^*)=-\vone^Tu^*+\mu\sum_{j=1}^mp_j(u^*)\leq -(1-\epsilon/2)\opt+\mu m (\frac{\epsilon}{mn})^2\leq -(1-\epsilon)\opt$.
\item By contradiction, suppose $f_\mu(x)<-(1+\epsilon)\opt$, since $f_\mu(x)>-\vone^Tx$, we must have $\vone^Tx>(1+\epsilon)\opt$. Suppose $\vone^Tx=(1+v)\opt$ for some $v>\epsilon$. There must exits a $j$, such that $(Ax)_j>1+v$. Then we have $p_j(x)>\exp(v/\mu)=((\frac{mn}{\epsilon})^4)^{v/\epsilon}$, which implies 
\[
f_\mu(x)\geq -(1+v)\opt + \mu p_j(x)\geq \frac{\epsilon}{4\log(mn/\epsilon)}((\frac{mn}{\epsilon})^4)^{v/\epsilon}-(1+v)\opt>0  ,
\]
since $\opt\leq n$, and $v>\epsilon$. This gives a contradiction.
\item The $x_0$ we use satisfies $Ax_0-\vone\leq -\epsilon/2 - \vone$, thus
\[
f_\mu(x_0)=\mu\sum_j p_j(x_0)-\vone^Tx_0\leq \frac{\mu m}{(nm)^2}-\frac{1-\epsilon/2}{n}\leq -\frac{1-\epsilon}{n}  .
\]
\item By contradiction, suppose there is some $j$ such that $(Ax)_j-1\geq \epsilon$. Let $v>\epsilon$ be the smallest $v$ such that  $Ax\leq (1+v)\opt$, and denote $j$ the constraint that has $(Ax)_j-1=v$. We must have $\vone^Tx\leq (1+v)\opt$ by definition of $\opt$. Then
\[
f_\mu(x)\geq \mu p_j(x)-(1+v)\opt\geq \frac{\epsilon}{4\log(mn/\epsilon)}((\frac{mn}{\epsilon})^4)^{v/\epsilon}-(1+v)\opt>0,
\]
which gives a contradiction.
\item By the above part, $f_\mu(x)\leq -(1-O(\epsilon))\opt\leq 0$ suggests $\frac{x}{1+\epsilon}$ is feasible. Furthermore, $-\vone^Tx < f_\mu(x) \leq -(1-O(\epsilon))\opt$ gives $\vone^Tx\geq (1-O(\epsilon))\opt$, thus $\vone^T\frac{x}{1+\epsilon} \geq (1-O(\epsilon))\opt$ is approximately optimal.
\item This is by straightforward computation.
\end{enumerate}
\end{proof}

\mirrorstep*

\begin{proof}
Since the function $V_{x^{(k)}}(z)$, the dot product and the constraint $z\geq 0$ are all coordinate-wise separable, we look at each coordinate independently. 
Thus we only need to check
\[
x_{k+1}^{(t)}[i]=\argmin_{z[i]\geq 0}\{(z[i]\log\frac{z[i]}{x_k[i]}+x_k[i]-z[i])+\alpha \xi_{k}^{(t)}[i](z[i]-x_k[i]) \} .
\]
This univariate function being optimized is convex and has a unique minimizer. We find it by taking the derivative to get
\[
\log\frac{z[i]}{x_k[i]}+\alpha \xi_{k}^{(t)}[i] = 0   ,
\]
which gives $x_{k+1}^{(t)}[i]\myeq z[i]=x_k[i]\exp(-\alpha \xi_{k}^{(t)}[i])$.
\end{proof}

\end{appendices}

\bibliographystyle{plain}
\bibliography{PP}

\end{document}